\documentclass[a4paper,10pt]{article}

\usepackage{pfFormACMSmallV13}
\usepackage{pfResV8}
\usepackage{hyperref}
\usepackage{algorithmic}
\usepackage{tikz}
\usepackage{authblk}
\usepackage{paralist}
\usetikzlibrary{arrows,automata}
\tikzset{every state/.style={minimum size=40pt}}

\graphicspath{{Figures/}}

\title{Corrections to \emph{A Menagerie of Timed Automata}}
\author[1,2]{Jeroen J.A. Keiren}
\affil[1]{Open University in The Netherlands}
\affil[2]{Radboud University, Nijmegen, The Netherlands}
\author[3]{Peter Fontana}
\author[3]{Rance Cleaveland}
\affil[3]{University of Maryland, USA}

\date{}

\newcommand{\CX}{\mathit{CX}}
\newcommand{\TA}{\mathit{TA}}
\newcommand{\TS}{\mathit{TS}}
\newcommand{\INV}{\mathit{INV}}
\newcommand{\URG}{\mathit{URG}}
\newcommand{\PUR}{\mathit{PUR}}

\begin{document}

\maketitle

\begin{abstract}
This note corrects a technical error in the ACM Computing Surveys paper mentioned in the title.  The flaw involved constructions for showing that timed automata with urgent locations have the same expressiveness as timed automata that allow false location invariants.  Corrected constructions are presented in this note, and the affected results are reproved.
\end{abstract}

\section{Introduction}

This note corrects a technical shortcoming in the \emph{ACM Computing Surveys} paper \emph{A Menagerie of Timed Automata} published 3 January 2014 \cite{fontana-a-menagerie-2014} (DOI: \url{http://dx.doi.org/10.1145/2518102}). This note often refers to that paper for details, to avoid repeating a large part of it. It is therefore advisable to have a copy of \cite{fontana-a-menagerie-2014} for reference.

That paper developed a unified framework for so-called \emph{timed automata}, which extend traditional finite-state machines with real-valued clock variables.  The states, or \emph{locations}, in these timed automata are equipped with \emph{location invariants} describing a property that must hold of the clock variables in order for control to remain within the given location.  Some accounts of timed automata do not allow control to change into locations whose invariants are false; others permit this behavior, in which case time is not permitted to advance until control exits from the location.

In the baseline version of timed automata considered in the original paper, transitions were not allowed into states whose location invariants would be violated by such a transition.  However, states whose invariants were violated were allowed to engage in action transitions.  Such states were not reachable from initial states, with the following exception:  initial states themselves were allowed to have invariant violations in them.  Specifically, if $l_0 \in L_0$ and $\nu_0 \not\models I(l_0)$, $(l_0,\nu_0)$ was nevertheless allowed to be an initial state.

This decision makes some of the semantic conversions contained in the paper (and this note) easier, but it also may be viewed as being inconsistent with our treatment of invariants in non-initial states.  In particular, one might wish for the following to be true in each reachable state $(l,\nu)$ in a timed automaton $\TA$:  $\nu \models I(l)$.  This can fail to hold for initial states in $\TA$, as illustrated in the following example.
\begin{example}\label{ex:basicta}
Consider the following timed automaton, in which invariant $x > 1$ in the initial location $l_0$ is not satisfied by the clock valuation $[x:=0]$ assigning $0$ to the only clock $x$.
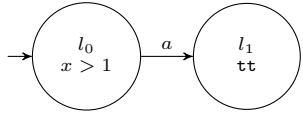
\begin{center}
  \begin{tikzpicture}[>=stealth',draw,node distance=60pt,initial text=]
  \scriptsize
  \node [state,initial left] (l0) {$\begin{array}{c}l_0\\x>1\end{array}$};
  \node [state,right of=l0] (l1) {$\begin{array}{c}l_1\\\lgtrue{}\end{array}$};
  \path[->] (l0) edge node[above] {$a$} (l1);
  \end{tikzpicture}
\end{center}
According to the baseline semantics, only states are entered in which the invariant is satisfied.  However, it does not disallow the situation in which the invariant is violated initially.  Hence, the baseline semantics allows the executions $(l_0, [x := 0]) \ttrans{a} (l_1, [x := 0]) \ttrans{\delta} (l_1, [x := \delta])$ for all $\delta \in \Rr^{\geq 0}$, even though initially $I(l_0)$ is not satisfied.  $\blacksquare$
\end{example}

In two places in the original paper, an implicit assumption was made that when a location invariant is violated in a starting location, no behavior is possible in that location.  As illustrated in the previous example, this assumption is at odds with assumptions made elsewhere in that paper.  As a result, two of the semantic transformations given in the paper do not correctly handle invariant violations in initial locations.  This note explains how the transformations may be modified so initial invariant violations are handled consistently.  The parts of the paper that this corrigendum addresses involve Section 5.1, where transformations in question are defined, and associated appendices, where proofs are given.

To facilitate the description of the timed-automaton conversions below, we recall \cite[Definitions 3.1 (clock constraints) and 3.2 (timed automaton)]{fontana-a-menagerie-2014}.

\setcounter{section}{3}  

\begin{definition}[Clock constraint $\phi \in \Phi(\CX)$ from  \cite{fontana-a-menagerie-2014}] Given a nonempty finite set of clocks $\CX = \mset{x_1, x_2, \ldots, x_n}$ and $c \in \Zz^{\geq 0}$ (a non-negative integer), a \emph{clock constraint $\phi$} may be constructed using the following grammar:
\begin{equation*}
\phi::= x_i < c \ | \ x_i \leq c \ | \ x_i > c \ | \ x_i \geq c \ | \ \phi \lgcand \phi
\end{equation*}
$\Phi(\CX)$ is the set of all possible clock constraints over $\CX$. We also use the following abbreviations: true (\lgtrue) for $x_1 \geq 0$, false (\lgfalse) for $x_1 < 0$, and $x_i = c$ for $x_i \leq c \lgcand x_i \geq c$. 
\label{def:cxcons}
\end{definition}

\begin{definition}[Timed automaton from  \cite{fontana-a-menagerie-2014}]
A \emph{timed automaton} $TA = (L, L_0, L_u, \Sigma, \CX, I, E)$ is a tuple where:
\begin{compactitem}
\item $L$ is the finite set of \emph{locations}. 
\item $L_0 \subseteq L$ is the nonempty set of \emph{initial locations}.
\item $L_u \subseteq L$ is the set of \emph{urgent locations}.
\item $\Sigma$ is the finite set of \emph{action symbols}. 
\item $\CX$ is the nonempty finite set of \emph{clocks} ($\CX = \mset{x_1, x_2, \ldots, x_n}$).
\item $\nmfunc{I}{L}{\Phi(\CX)}$ gives a clock constraint for each location $l$.  $I(l)$ is referred to as the \emph{invariant} of $l$.
\item $E \subseteq L \times \Sigma \times \Phi(\CX) \times 2^{\CX} \times L$ is the set of \emph{edges}.  In an edge $e = (l, a, \phi, \lambda, l')$ from $l$ to $l'$ with action $a$, $\phi \in \Phi(\CX)$ is the \emph{guard} of $e$, and $\lambda \in 2^{\CX}$ represents the set of clocks to \emph{reset} to $0$ when the edge is executed. 
\end{compactitem}
\label{def:timedaut}
\end{definition}

\setcounter{section}{1}  

One assumption made in that paper, and in others involving timed automata, is that $\Sigma \cap \Rr^{\geq 0} = \emptyset$; in other words, $\Sigma$ does not include any non-negative real numbers, which are reserved for use in the semantics of these automata.

In \cite{fontana-a-menagerie-2014}, timed automata are given a baseline semantics in the form of a translation function that maps a timed automaton $\TA = (L, L_0, L_u, \Sigma, \CX, I, E)$ to a timed transition system $\TS(\TA) = (Q,Q_0, \Delta(\Sigma),\longrightarrow)$, where the set of states $Q$ consists of pairs of automaton locations and clock assignments (i.e.\/ mappings of clocks to non-negative real numbers), $Q_0 \subseteq Q$ is the set of initial states, $\Delta(\Sigma) = \Sigma \cup \Rr^{\geq 0}$ is the set of transition labels (actions or time elapses), and $\ttrans{} \,\subseteq Q \times \Delta(\Sigma) \times Q$ is the transition relation.\footnote{The original paper does not introduce the notation $\Delta(\Sigma)$; we do so here for improved clarity.}   The details of this construction may be found in \cite[Definition 3.7]{fontana-a-menagerie-2014}, and forbids transitions into states $(l, \nu)$ where $\nu \not\models I(l)$ (i.e.\/ the clock assignment violates the location invariant of location $l$. Satisfaction of clock valuations, $\models$, is made precise in the usual fashion).

In \cite[Section 5.1]{fontana-a-menagerie-2014}, a semantic variant of timed automata is considered that weakens the restriction on transitions into transition-system states $(l,\nu)$ for which $\nu \not\models I(l)$.  Specifically, the new semantics associates a transition system $\TS'(\TA) = (Q,Q_0,\Delta(\Sigma),\ttrans{})$ with $\TA$, where $Q, Q_0$ and $\Delta(\Sigma)$ retain the definitions above and $\ttrans{}$ is redefined as specified in the lower part of \cite[page 20 in Section 5.1]{fontana-a-menagerie-2014}.  (The notation $TS'$ is not used in the paper, but is introduced here to simplify the presentation.)

Two transformations are then given in \cite[Section 5.1]{fontana-a-menagerie-2014}, $\INV$ and $\URG$, that are intended to have the following properties.  Given a timed automaton $\TA = (L, L_0, \emptyset, \Sigma, \CX, I, E)$ with an empty set of urgent locations, $\INV(\TA)$ has the property that $\TS'(\TA)$ and $\TS(\INV(\TA))$ are semantically indistinguishable, in a precisely defined sense \cite[Theorem 5.5]{fontana-a-menagerie-2014}.  That is, $\TA$ interpreted in a semantics in which action transitions are allowed in states with location-invariant violations is equivalent to $\INV(\TA)$ interpreted in our baseline semantics.  Similarly, given a baseline timed automaton $\TA$, $\URG$ has the property that $\TS(\TA)$ and $\TS'(\URG(\TA))$ are appropriately related \cite[Theorem 5.6]{fontana-a-menagerie-2014}.

The constructions $\INV$ and $\URG$ are the ones that this note redefines to eliminate the issues with violated invariants.  The modified conversions are given, and \cite[Theorems 5.5 and 5.6]{fontana-a-menagerie-2014} reproved.  The note then concludes with a new construction showing how violated invariants may be eliminated entirely from the baseline formalism.

\section{Conversion $\INV$ (to Baseline Version)}
\label{s:invconv1}

Before we continue defining the translation $\INV$ we first illustrate the problem with the translation in the original paper.
\begin{example}\label{ex:invbug}
The timed automaton from Example~\ref{ex:basicta} is translated into the following timed automaton using the original translation $\INV$.  Locations $l'_{0,u}$ and $l'_{1,u}$ are urgent.
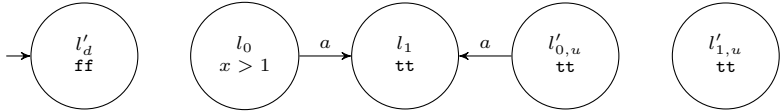
\begin{center}
  \begin{tikzpicture}[>=stealth',draw,node distance=60pt,initial text=]
  \scriptsize
  \node [state,initial left] (l'd) {$\begin{array}{c}l'_d\\\lgfalse{}\end{array}$};
  \node [state,right of=l'd] (l0) {$\begin{array}{c}l_0\\x>1\end{array}$};
  \node [state,right of=l0] (l1) {$\begin{array}{c}l_1\\\lgtrue{}\end{array}$};
  \node [state,right of=l1] (l'0u) {$\begin{array}{c}l'_{0,u}\\\lgtrue{}\end{array}$};
  \node [state,right of=l'0u] (l'1u) {$\begin{array}{c}l'_{1,u}\\\lgtrue{}\end{array}$};
  \path[->] (l0) edge node[above] {$a$} (l1)
            (l'0u) edge node[above] {$a$} (l1);
  
  \end{tikzpicture}
\end{center}
The timed transition system underlying this automaton, using the baseline semantics, does not allow any transitions from the initial state, whereas the timed automaton from Example~\ref{ex:basicta}, interpreted in the unsatisfied invariants semantics, allows an $a$ transition from the initial state.  Therefore, the translation does not preserve the semantics of the original timed automaton. $\blacksquare$
\end{example}
The original version of $\INV$ incorrectly introduced dead locations $l_d$ for initial locations whose invariants are not satisfied by the initial clock valuation.  To correct the definition,
let $TA = (L, L_0, \emptyset, \Sigma, \CX, I, E)$ be a timed automaton with an empty set of urgent locations, and let $L_u = \mset{l_u \ | \ l \in L}$ be a fresh set of locations with the property that $L_u \cap L = \emptyset$ and $l_u \neq l'_u$ if $l \neq l'$.   Also let $\nu_0$ be the clock valuation assigning $0$ to every clock in $\CX$.  Finally, we recall \cite[Definition 5.2]{fontana-a-menagerie-2014} from the original paper, which introduces of $\texttt{resetPred}(\phi, \lambda)$, where $\phi$ is a clock constraint and $\lambda \subseteq \CX$ is a set of clocks to be reset.  The constraint $\texttt{resetPred}(\phi, \lambda)$ may be viewed as the weakest precondition of $\phi$ with respect to the simultaneous assignment of each clock in $\lambda$ to $0$; it is the weakest property $\phi'$ such that if $\nu \models \phi'$, then $\nu[\lambda := 0] \models \phi$.
We now redefine 
$\INV(\TA) = (L', L'_0, L'_u, \Sigma, \CX, I', E')$ as follows.
	\begin{compactitem}
		\item $L' = L \cup L_u$.
		\item $L'_0 = \mset{l \in L_0 \ | \ \nu_0 \models I(l)} \cup \mset{l_u \in L_u \ | \ l \in L_0 \wedge \nu_0 \not\models I(l)}$. 
		\item $L'_u = L_u$.
		\item $I'(l') =
		\left\{
		\begin{array}{lp{1.5in}}
		I(l')		& if $l' \in L$\\
		\lgtrue		& otherwise (i.e.\/ $l' \in L_u$)
		\end{array}
		\right.
		$
		\item For each edge $(l,a,\phi,\lambda, l') \in E$,  $E'$ includes the following four edges, where $\phi_1 = \phi \wedge \texttt{resetPred}(I(l'),\lambda)$ and $\phi_2 = \phi \wedge \neg \texttt{resetPred}(I(l'), \lambda)$.
\begin{align*}
(l, a, \phi_1, \lambda, l'), \\
(l, a, \phi_2, \lambda, l'_{u}),\\
(l_{u}, a, \phi_1, \lambda, l'),\\
(l_{u}, a, \phi_2, \lambda, l'_{u}).
\end{align*}
Disjunctive guard constraints may arise from negating \texttt{resetPred($I(l'), \lambda$)}. Following the process used in \cite[Section 4.1]{fontana-a-menagerie-2014} of the original paper, any disjunctive guard constraint is eliminated by converting the edge with such a constraint to a set of edges. 
\end{compactitem}

The key difference in the redefinition of $\INV$ involves $L'_0$.  In the original construction, $L_0'$ was incorrectly taken to include a set of dead locations $L_d$ to represent those initial locations whose invariants were violated by the initial clock assignment $\nu_0$.
In the new construction, initial locations $l \in L_0$ that are violated by the initial clock assignment $\nu_0$ are replaced in $L'_0$ by their urgent versions $l_u$.
\begin{example}\label{ex:invfixed}
The timed automaton from Example~\ref{ex:basicta}, interpreted in the unsatisfied invariants semantics, is translated into the following timed automaton in the baseline semantics using the fixed translation $\INV$.  Locations $l'_{0,u}$ and $l'_{1,u}$ are urgent.
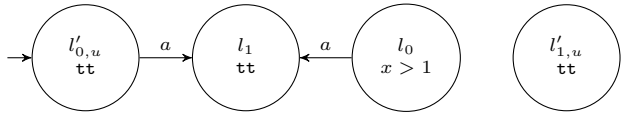
\begin{center}
  \begin{tikzpicture}[>=stealth',draw,node distance=60pt,initial text=]
  \scriptsize
  \node [state,initial left] (l'0u) {$\begin{array}{c}l'_{0,u}\\\lgtrue{}\end{array}$};
  \node [state,right of=l'0u] (l1) {$\begin{array}{c}l_1\\\lgtrue{}\end{array}$};
  \node [state,right of=l1] (l0) {$\begin{array}{c}l_0\\x>1\end{array}$};
  \node [state,right of=l0] (l'1u) {$\begin{array}{c}l'_{1,u}\\\lgtrue{}\end{array}$};
  \path[->] (l0) edge node[above] {$a$} (l1)
            (l'0u) edge node[above] {$a$} (l1);
  
  \end{tikzpicture}
\end{center}
It is not hard to see that the underlying timed transition system in the baseline semantics, when restricted to reachable states, is the same as that of the original timed automaton in the unsatisfied invariants semantics. $\blacksquare$
\end{example}
We now state and prove \cite[Theorem 5.5]{fontana-a-menagerie-2014} from the original paper.

\newcounter{oldsection}
\setcounter{oldsection}{\value{section}}
\setcounter{section}{5}
\setcounter{theorem}{4}
\begin{theorem}
Let $\TA = (L, L_0, \emptyset, \Sigma, \CX, I, E)$ be a timed automaton with an empty set of urgent locations.  Then the reachable subsystems of $\TS'(\TA)$ and $\TS(\INV(\TA))$ are isomorphic, \emph{i.e.} $\TS'(\TA)) \cong_{r} \TS(\INV(\TA))$. 
\end{theorem}

\setcounter{section}{\value{oldsection}}

\begin{proof}[Proof of Theorem 5.5]  

Given transitition system $T = (Q,Q_0,\Delta(\Sigma),\ttrans{})$, define the reachable state space of $T$, $R(T) \subseteq Q$,  to be the smallest subset of $Q$ satisfying the following:
\begin{compactitem}
\item $Q_0 \subseteq R(T)$;
\item if $q \in R(T)$ and $q \ttrans{\alpha} q'$ for some $\alpha \in \Delta(\Sigma)$ then $q' \in R(T)$.
\end{compactitem}
In what follows, we sometimes abuse notation and write $R(T)$ for the transition system $(R(T), Q_0, \Delta(\Sigma), \ttrans{})$.

Let $\TS'(\TA) = (Q_1, Q_{0,1}, \Delta(\Sigma), \ttrans{}_1)$, and $\TS(\INV(\TA)) = (Q_2, Q_{0,2}, \Delta(\Sigma),$ $\ttrans{}_2)$.  To prove the theorem we must give an isormorphism $f$ from $R(\TS'(\TA))$ to $R(\TS(\INV(\TA)))$; specifically $\nmfunc{f}{R(\TS'(\TA))}{R(\TS(\INV(\TA)))}$ must have the following properties.
\begin{compactenum}
\item $f$ is one-to-one.
\item $f(Q_{0,1}) = Q_{0,2}$.
\item For every $q, q' \in R(\TS'(\TA))$ and $\alpha \in \Delta(\Sigma)$, $q \ttrans{\alpha}_1 q'$ iff $f(q) \ttrans{\alpha}_2 f(q')$.
\end{compactenum}
We define $\nmfunc{f}{Q_1}{Q_2}$ as follows, then show it is an isomorphism from $R(\TS'(\TA))$ to $R(\TS(\INV(\TA)))$.  So consider
$$
f((l,\nu)) =
\begin{cases}
(l, \nu), &\text{ if } \nu \models I(l) \\
(l_u, \nu) &\text{ otherwise.}
\end{cases}
$$
We now prove that $f$ has the necessary properties.

\paragraph{$f$ is one-to-one.}  Suppose $f((l,\nu)) = f((l',\nu')) = (l'',\nu'')$; we must show that $l = l'$ and $\nu = \nu'$.  From the definition of $f$ it follows that if $f((l,\nu)) = (l'',\nu'')$ then $\nu = \nu''$; hence, $\nu = \nu' = \nu''$.  Moreover, if $\nu \models I(l)$ then $l'' = l$, whence it must be the case in $l = l'$.  Finally, if $\nu \not\models I(l)$ then $l'' = l_u$, and again it follows that $l = l'$.

\paragraph{$f(Q_{0,1}) = Q_{0,2}$.}  Suppose $(l,\nu) \in Q_{0,1}$; we must show that $f((l,\nu)) = (l',\nu') \in Q_{0,2}$.  First, note that $l \in L_0$ and $\nu = \nu_0$, and thus $\nu' = \nu_0$ by the definition of $f$.  By the definition of $\INV$, it also follows that $l' = l$ if $\nu_0 \models I(l)$ and $l' = l_u$ otherwise.  In either case, $l' \in L_0'$ and $(l',\nu') \in Q_{0,2}$.  Now suppose $(l',\nu') \in Q_{0,2}$; we must give $(l,\nu) \in Q_{0,1}$ such that $f((l,\nu)) = (l',\nu')$.  As before, $\nu' = \nu_0$ by definition of $Q_{0,2}$.  Now either $l' \in L_0$, meaning $\nu_0 \models I(l')$ and thus $f((l',\nu')) = (l',\nu)$, or $l' = l_u$ for some $l \in L_0$ with $\nu_0 \not\models I(l)$; in this case, $f((l,\nu_0)) = (l',\nu')$, with $(l,\nu_0) \in Q_{0,1}$.

\paragraph{$(l,\nu) \ttrans{\alpha}_1 (l',\nu')$ iff $f((l,\nu)) \ttrans{\alpha}_2 f((l',\nu'))$.}  There are two cases to consider:  $\alpha = \delta$ for some $\delta \in \Rr^{\geq 0}$, or $\alpha \in \Sigma$.  So suppose $\alpha = \delta$ for some $\delta \geq 0$.  Now, $(l,\nu) \ttrans{\delta}_1 (l',\nu')$ iff $l = l'$, $\nu' = \nu + \delta$, and for all $k$ such that $0 \leq k \leq \delta$, $\nu+k \models I(l)$.  This in turn holds iff $l = l'$,   $f((l,\nu)) = (l,\nu)$, $f((l',\nu')) = (l',\nu')$, and $f((l,\nu)) \ttrans{\delta}_2 f((l',\nu'))$.

Now consider the case where $\alpha = a \in \Sigma$, and suppose $(l,\nu) \ttrans{a}_1 (l',\nu')$ is an action transition, meaning there is an edge $e = (l,a,\phi,\lambda,l') \in E$ such that $\nu \models \phi$ and $\nu' = \nu[\lambda := 0]$.  We must show that $f((l,\nu)) \ttrans{a}_2 f((l',\nu'))$.  There are four cases to consider.
\begin{compactitem}
\item $\nu \models I(l)$ and $\nu' \models I(l')$.
In this case, $f((l,\nu)) = (l,\nu)$, $f((l',\nu')) = (l',\nu')$, and $\nu \models \phi_1$ as defined in $\INV$.  By our conversion, we have the edge $e = (l, a, \phi \cap \texttt{resetPred($I(l'), \lambda$)}, \lambda, l')$ in $\INV(\TA)$ and, $f((l', \nu[\lambda := 0])) = (l', \nu[\lambda :=0])$. Since we know $\nu[\lambda := 0] \models I(l')$ and $\nu \models \phi$, by Corollary B.5, we know $\nu \models \phi \cap \texttt{resetPred($I(l'), \lambda$)}$. Therefore, $\INV(\TA)$ has the transition $f((l, \nu)) \ttrans{a} f((l', \nu[\lambda := 0]))$. 

\item $\nu \models I(l)$ and $\nu' \not\models I(l')$.
In this case, $f((l,\nu)) = (l,\nu)$, $f((l',\nu')) = (l'_u,\nu')$, and $\nu \models \phi_2$ as defined in $\INV$.  By our conversion, we use the edge $e_{u} = (l_{u}, a, \phi \cap \texttt{resetPred($I(l'), \lambda$)}, \lambda, l')$ in $\INV(\TA)$. Otherwise, the proof is the same as the previous case's. It therefore follows that $f((l,\nu)) \ttrans{a}_2 f((l',\nu'))$, since the edge $(l,a,\phi_2,\lambda,l_u')$ is in the edge set of $\INV(\TA)$.

\item $\nu \not\models I(l)$ and $\nu' \models I(l')$.
In this case, $f((l,\nu)) = (l_u,\nu)$, $f((l',\nu')) = (l',\nu')$, and $\nu \models \phi_1$ as defined in $\INV$.  By our conversion, $f((l', \nu[\lambda:=0])) = (l'_{u}, \nu[\lambda := 0])$. Since $l'_u$ is the urgent copy of $l'$, we know $\nu[\lambda := 0] \models I(l'_u)$.  Since $\nu \models \phi$, by Corollary B.5, we know that $\nu \models \phi \cap \neg \texttt{resetPred($I(l'), \lambda$)}$. By the definition of the transition system semantics, $\INV(\TA)$ has the transition $f((l, \nu)) \ttrans{a} f(l', \nu[\lambda := 0])$. 

\item $\nu \not\models I(l)$ and $\nu' \not\models I(l')$.
In this case, $f((l,\nu)) = (l_u,\nu)$, $f((l',\nu')) = (l'_u,\nu')$, and $\nu \models \phi_2$ as defined in $\INV$.   By our conversion, we  use the edge $e_{u} = (l_{u}, a, \phi \cap \neg\texttt{resetPred($I(l'), \lambda$)}, \lambda, l_{u}')$ in $\INV(\TA)$. Otherwise, the proof is the same as the previous case's. It therefore follows that $f((l,\nu)) \ttrans{a}_2 f((l',\nu'))$, since the edge $(l_u,a,\phi_2,\lambda,l'_u)$ is in the edge set of $\INV(\TA)$.
\end{compactitem}
For the converse, we assume that $f((l,\nu)) \ttrans{a}_2 f((l',\nu'))$ and must show that $(l,\nu) \ttrans{a}_1 (l',\nu')$.  The argument follows the lines above and relies on a case analysis of which of the four types of edges in $\INV(\TA)$ supports the conclusion that $f((l,\nu)) \ttrans{a}_2 f((l',\nu'))$.  The details are omitted (one may wish to use the definition of $f^{-1}$ when proving the converse).

\paragraph{$(l,\nu) \in R(\TS'(\TA))$ iff $f((l,\nu)) \in R(\TS(\INV(\TA)))$.}
This can be proved by induction on the definition of $R(\cdot)$ and is a consequence of the fact that $f(Q_{0,1}) = Q_{0,2}$ and that $(l,\nu) \ttrans{\alpha}_1 (l',\nu')$ iff $f((l,\nu)) \ttrans{\alpha}_2 f((l',\nu'))$.  Alternatively, to show this, suppose we have a state $(l_{inv}, \nu_{inv})$ in $\INV(\TA)$ that is not mapped to by $f$. We claim  $(l_{inv}, \nu_{inv})$ is not reachable from an initial state. By the definition of $f$, if $l_{inv}$ is not an urgent copy and $\nu_{inv} \models I(l_{inv})$, then $(l_{inv}, \nu_{inv})$ is covered by $f$. Likewise, if $l_{inv}$ is an urgent copy location $l_{u}$ and $(l, \nu_{inv}) \not\models I(l)$, then $(l_{inv}, \nu_{inv})$ is covered by $f$. If $l_{inv}$ is an urgent copy location $l_{u}$ and $(l, \nu_{inv}) \models I(l)$, then by the construction of $\INV(\TA)$, this state is not reachable. If $l_{inv}$ is not an urgent copy and $\nu_{inv} \not\models I(l_{inv})$, then by the semantics of $\INV(\TA)$, $(l_{inv},\nu_{inv})$ is only reachable if and only it is an initial state. Furthermore, by construction, only urgent initial locations do not satisfy their invariant. Thus, applying $f$ to a reachable state in the original timed automaton results in a reachable state in the converted timed automaton.

\paragraph{$f$ is an isomorphism from $R(\TS'(\TA))$ to $R(\TS(\INV(\TA)))$.} 
This conclusion is a consequence of the previous facts.  Since $f$ is one-to-one on $Q_1$, it is one-to-one when restricted to $R(\TS'(\TA)) \subseteq Q_1$.  Moreover, as $(l,\nu) \in R(\TS'(\TA))$ iff $f((l,\nu)) \in R(\TS(\INV(\TA)))$, it follows that $f(R(\TS'(\TA)) = R(\TS(\INV(\TA)))$, meaning $f$ when restricted to $R(\TS'(\TA))$ is onto with respect to $R(\TS(\INV(\TA)))$.  Hence $f$ is a bijection from $R(\TS'(\TA))$ to $R(\TS(\INV(\TA)))$ that preserves start states and the transition relation, and is therefore the required isomorphism.
\end{proof}

\section{Conversion $\URG$ (from Baseline Version)}
\label{s:invconv2}

The original definition of $\URG$ also erroneously introduced dead locations for locations in $\TA$ that were not satisfied by the initial clock assignment.
\begin{example}\label{ex:urgbug}
The timed automaton from Example~\ref{ex:basicta} is translated into the following timed automaton using the original translation $\URG$.
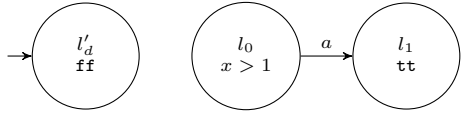
\begin{center}
  \begin{tikzpicture}[>=stealth',draw,node distance=60pt,initial text=]
  \scriptsize
  \node [state,initial left] (l'd) {$\begin{array}{c}l'_d\\\lgfalse{}\end{array}$};
  \node [state,right of=l'd] (l0) {$\begin{array}{c}l_0\\x>1\end{array}$};
  \node [state,right of=l0] (l1) {$\begin{array}{c}l_1\\\lgtrue{}\end{array}$};
  \path[->] (l0) edge node[above] {$a$} (l1);
  
  \end{tikzpicture}
\end{center}
Using the unsatisfied-invariants semantics, the timed transition system underlying this automaton does not allow any transitions from the initial state, whereas the timed automaton from Example~\ref{ex:basicta}, interpreted in the baseline semantics, allows an $a$ transition from the initial state. The translation thus does not preserve the semantics. $\blacksquare$
\end{example}
We now detail the modified construction $\URG$, which converts timed automata from our baseline formalism into automata permitting transitions into states with location invariants, and reprove the associated correctness result.

Specifically, let $\TA = (L, L_0, L_u, \Sigma, \CX, I, E)$ be a timed automaton.  We wish to define another timed automaton, $\URG(\TA) = (L', L'_0, \emptyset, \Sigma, \CX, I', E')$, with an empty set of urgent locations, so that $\TS(\TA)$ and $\TS'(\URG(\TA))$ are isomorphic in an appropriate sense.  $\URG(\TA)$ may be given as follows.
\begin{compactitem}
\item $L' = L$
\item $L'_0 = L_0$
\item $I'(l) =
\begin{cases}
I(l)		& \text{if } l \not\in L_u \\
\lgfalse{} & \text{otherwise.}
\end{cases}
$
\item $E' = \mset{(l, a, \phi \wedge \texttt{resetPred}(I(l'), \lambda), \lambda, l') \ | \ (l, a, \phi, l, \lambda, l') \in E}$
\end{compactitem}

\begin{example}\label{ex:urgfixed}
The timed automaton from Example~\ref{ex:basicta} is not modified by the new translation $\URG$.  We have already observed that the underlying timed transition system is the same for the baseline and the unsatisfied invariants semantics. $\blacksquare$
\end{example}
We now restate and reprove \cite[Theorem 5.6]{fontana-a-menagerie-2014}.

\setcounter{oldsection}{\value{section}}
\setcounter{section}{5}
\setcounter{theorem}{5}
\begin{theorem}
Let $TA = (L, L_0, L_u, \Sigma, \CX, I, E)$ be a timed automaton.  Then $\TS(\TA)$ and $\TS'(\URG(\TA))$ are isomorphic, \emph{i.e.} $\TS'(\TA) \cong \TS(\INV(\TA))$. 
\end{theorem}

\setcounter{section}{\value{oldsection}}

\begin{proof}[Proof of Theorem 5.6]  

It should first be noted that both transition systems $\TS(\TA)$ and $\TS'(\URG(\TA))$ have the same set of states and initial states.  Unlike the proof of correctness of $\INV$, in this case full isomorphism of the timed transition systems of $\TS(\TA) = (Q,Q_0,\Delta(\Sigma), \ttrans{}_1)$ and $\TS'(\URG(\TA)) = (Q, Q_0, \Delta(\Sigma), \ttrans{}_2)$ can be established. Consider the function $f$:
\begin{align*}
\nmfunc{f&}{Q}{Q}  \\
f&\bigl((l, \nu)\bigr) = (l,\nu)
\end{align*}
or the identity function.  
We must show that that $f$ is an isomorphism from $\TS(\TA)$ to $\TS'(\URG(\TA))$.  That $f$ is a bijection, and that $f(Q_0) = Q_0$, follow from $f$ being the identity function.  It remains to show that $(l, \nu) \ttrans{\alpha}_1 (l',\nu')$ iff $f((l,\nu)) \ttrans{\alpha}_2 f((l',\nu'))$ for all $\alpha \in \Delta(\Sigma)$.  There are two cases to consider.
\begin{compactitem}
\item $(l,\nu)\ttrans{\delta}_1 (l',\nu')$ some $\delta \geq 0$.  This happens iff $l \not\in L_u$, $l = l'$, $\nu' = \nu + \delta$, and for all $k$ such that $0 \leq k \leq \delta$, $\nu + k \models I(l)$, which holds iff $f((l,\nu)) = (l, \nu) \ttrans{\delta}_2 (l',\nu') = f((l',\nu'))$.
\item $(l,\nu)\ttrans{a}_1 (l',\nu')$ some $a \in \Sigma$.  This happens iff there exists an edge $(l, a, \phi, \lambda, l') \in E$ such that $\nu \models \phi$, $\nu' = \nu[\lambda :=0]$, and $\nu' \models I(l')$, which in turn is logically equivalent to asserting the existence of an edge $(l, a, \phi, \lambda, l') \in E$ such that $\nu' = \nu[\lambda :=0]$ and $\nu \models \phi \wedge \texttt{resetPred}(\phi,\lambda)$.  This holds iff there is an edge $(l, a, \phi \wedge \texttt{resetPred}(I(l'), \lambda), \lambda, l') \in E'$ iff (by Corollary B.5) $\nu \models \phi \cap \texttt{resetPred($I(l'), \lambda$)}$, which in turn holds iff $f((l, \nu)) = (l, \nu) \ttrans{a}_2 (l',\nu') = f((l',\nu'))$. 
\qedhere
\end{compactitem}
\end{proof}

\section{Initial Invariant Violations}
%

To conclude this note we introduce a new translation that shows that for every timed automaton $TA$, there is another timed automaton $\PUR(\TA)$ that is semantically equivalent, and in which each reachable state $(l,\nu)$ satisfies: $\nu \models I(l)$.  Let $\TA = (L, L_0, L_u, \Sigma, \CX, I, E)$, and define $L_B = \mset{l \in L_0 \ | \ \nu_0 \not\models I(l)}$.  Introduce a set $F_u = \mset{l_u \ | \ l \in L_B}$ of fresh locations with the property that $F_u \cap L = \emptyset$ and $l'_u \neq l''_u$ iff $l' \neq l''$.  Now consider $\PUR(\TA) = (L', L_0', L_u', \Sigma, \CX, I', E')$ given as follows.
\begin{compactitem}
\item $L' = L \cup F_u$
\item $L_0' = (L_0 - L_B) \cup F_u$
\item $L_u' = L_u \cup F_u$
\item $I'(l) =
\begin{cases}
I(l)		& \text{if } l \in L\\
\lgtrue{} & \text{otherwise (i.e. $l \in F_u$)}
\end{cases}
$
\item $E' = E \cup \mset{ (l'_u, a, \phi, \lambda, l'') \ | l' \in L_B \wedge (l', a, \phi, \lambda, l'') \in E\ }$
\end{compactitem}
In effect, $\PUR(\TA)$ creates fresh urgent initial locations for pre-existing ones whose invariants are not satisfied by the initial clock assignment, together with copies of edges from the old initial locations to these urgent ones.

\begin{example}\label{ex:pur}
Consider the timed automaton from Example~\ref{ex:basicta}, and observe that $L_B = \{ l_0 \}$, hence $F_u = \{ l_{0,u} \}$.  The automaton is translated into the following timed automaton using $\PUR$.  Location $l_{0,u}$ is urgent.
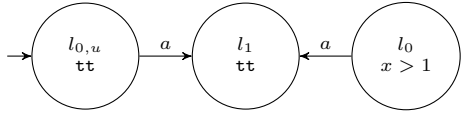
\begin{center}
  \begin{tikzpicture}[>=stealth',draw,node distance=60pt,initial text=]
  \scriptsize
  \node [state,initial left] (l0u) {$\begin{array}{c}l_{0,u}\\\lgtrue{}\end{array}$};
  \node [state,right of=l0u] (l1) {$\begin{array}{c}l_1\\\lgtrue{}\end{array}$};
  \node [state,right of=l0] (l0) {$\begin{array}{c}l_0\\x>1\end{array}$};
  
  \path[->] (l0u) edge node[above] {$a$} (l1)
            (l0) edge node[above] {$a$} (l1);
  
  \end{tikzpicture}
\end{center}
Observe that all \emph{reachable states} satisfy their invariant. $\blacksquare$
\end{example}
The following may now be proved about this translation.

\begin{theorem}
Let $TA$ be a timed automaton.  Then the following are true.
\begin{compactenum}
\item $\TS(\TA) \cong_r TS(\PUR(\TA))$.
\item Let $I'$ be the invariant mapping in $\TS(\PUR(\TA))$. Then for every $(l, \nu) \in R(\TS(\PUR(\TA))), \nu \models I'(l)$.
\end{compactenum}
\end{theorem}

\begin{proof}
Let $\TA = (L, L_0, L_u, \Sigma, \CX, I, E)$ and $\PUR(\TA) = (L', L'_0, L'_u, \Sigma, \CX,$ $I', E')$.  Also let $\TS(\TA) = (Q, Q_0, \Delta(\Sigma), \ttrans{})$ and $\TS(\PUR(\TA)) = (Q', Q'_0,$ $\Delta(\Sigma), \ttrans{}')$.  For Part 2 of the proof, it suffices to note that for every $(l'_0, \nu_0) \in Q'_0$, $\nu_0 \models I'(l'_0)$ by construction.  This, plus the fact that the definition of $\TS$ guarantees that if $(l,\nu) \ttrans{\alpha}' (l',\nu')$ then $\nu' \models I'(l')$, provides the desired result.

Now consider Part 1; we must devise an isomorphism $f$ from $R(\TS(\TA))$ to $R(\TS(\PUR(\TA)))$.  Define $\nmfunc{f}{Q}{Q'}$ as follows.
$$
f((l,\nu)) =
\begin{cases}
(l_u, \nu)	& \text{if } \nu = \nu_0 \text{ and } \nu_0 \not\models I(l) \text{ (i.e.\/ $l_u \in F_u$)}\\
(l, \nu)		& \text{otherwise.}
\end{cases}
$$
Following the proof given above of Theorem 5.5, to establish that $f$ is the desired isomorphism it suffices to argue that $f$ is one-to-one, that $f(Q_0) = Q'_0$, and that $(l, \nu) \ttrans{\alpha} (l',\nu')$ iff $f((l,\nu)) \ttrans{\alpha}' f((l',\nu'))$.  The first two of these follow immediately from the definitions of $f$ and set $L_0'$ of initial locations in $\PUR(\TA)$.

To show that $(l, \nu) \ttrans{\alpha} (l',\nu')$ iff $f((l,\nu)) \ttrans{\alpha}' f((l',\nu'))$, we consider two cases.  Suppose first that $f((l,\nu)) = (l, \nu)$.  In this case it must hold that $\nu' \models I(l')$, meaning that $f((l',\nu')) = (l',\nu')$, and the construction of $\PUR(\TA)$ guarantees that $f((l,\nu)) = (l,\nu) \ttrans{\alpha}' (l',\nu') = f((l',\nu'))$ iff $(l,\nu) \ttrans{\alpha} (l',\nu')$.

Now suppose that $f((l,\nu)) \neq (l,\nu)$.  This can only happen if $\nu = \nu_0$ and $\nu_0 \not\models I(l)$; in this case, $f((l,\nu)) = (l_u, \nu_0)$.  From the definition of $TS$, $(l, \nu) \ttrans{\alpha} (l,\nu')$ iff $\alpha \in \Sigma$ and there exists $(l,\alpha,\phi,\lambda,l') \in E$ with $\nu \models \phi$, $\nu' = \nu[\lambda := 0]$, and $\nu' \models I(l')$; in particular, $\alpha$ cannot be a delay event $\delta$.  Given the definition of $\PUR(\TA)$, this can happen iff there is an edge $(l_u,\alpha,\phi,\lambda,l') \in E'$, which in turn can hold iff $f((l,\nu)) = (l_u,\nu) \ttrans{\alpha}' (l',\nu') = f((l',\nu'))$ (the last equality follows from the fact that for any state $(l',\nu')$ with an incoming transition, $\nu' \models I(l')$ and hence $f((l',\nu')) = (l', \nu')$).
\end{proof}

\bibliographystyle{alpha}
\bibliography{CSURPfRc}

\end{document}